\documentclass[letterpaper, 10 pt, conference]{ieeeconf}  
\usepackage{cite}

\IEEEoverridecommandlockouts                              

\usepackage{graphics} 
\usepackage{float} 
\usepackage{epsfig} 
\usepackage{mathptmx} 
\usepackage{times} 
\usepackage{amsmath} 
\usepackage{amssymb}  
\usepackage{url}
\usepackage{cancel}
\usepackage{xcolor}
\usepackage{ifthen}
\usepackage{bm}
\usepackage{verbatim}
\usepackage{multicol}
\newboolean{showcomments}
\setboolean{showcomments}{true}

\newtheorem{proposition}{Proposition}

\newcommand{\dennice}[1]{\ifthenelse{\boolean{showcomments}}
{\textcolor{blue}{Dennice says: #1}}{}}
\newcommand{\addcites}[0]{\ifthenelse{\boolean{showcomments}}
{\textcolor{purple}{(add cite(s))}}{}}
\newcommand{\steve}[1]{\ifthenelse{\boolean{showcomments}}
{\textcolor{red}{Steve says: #1}}{}}
\begin{document}

\title{\LARGE \bf
From Cut-In to Rated: Multi-Region Floating Offshore Wind Farm Control for Secondary Frequency Regulation}

\author{Stephen Ampleman and Dennice F. Gayme
\thanks{This work was partially supported by the U.S. Department of Energy, Office of Science Energy Earthshot Initiative, under the Addressing Challenges in Energy—Floating Wind in a Changing Climate (ACE-
FWICC) Energy Earthshot Research Center. Partial support from the National Science Foundation (CBET 2401013 \& OISE 2330450) is also gratefully acknowledged. }
\thanks{Stephen Ampleman and Dennice F. Gayme are with the Department of Mechanical Engineering and the Ralph O'Connor Sustainable Energy Institute (ROSEI) at Johns Hopkins University.}
}

\maketitle
\thispagestyle{empty}
\pagestyle{empty}

\begin{abstract}
This paper describes a multi-region control framework for floating offshore wind farms. Specifically, we propose a novel generator torque controller that regulates rotor speed in Region 2, corresponding to wind speeds between the cut-in and rated values. In Region 3 (wind speeds at or above rated but below cut-out speed) we employ a PI-LQR for collective blade pitch. Control blending across the transitional wind speeds (Region 2.5) employs a sigmoid weighting function applied to the control variables. Two modeling paradigms are proposed for farm-level power tracking with rotor speed regularization: a nonlinear model predictive controller (NL-MPC) with a dynamic wake model, and a reduced order model predictive controller based on linear parameter varying turbine models with a time delay representation of wake advection (LPVTD-MPC). These approaches are evaluated over three wind inlet conditions using the PJM ancillary service certification criteria for participation in a secondary frequency regulation market. Results show that both approaches achieve scores of at least 89.9\% for the three different testing scenarios, which are well above the qualification threshold of 75\%. However, the LPVTD-MPC approach solves the problem in under half the time versus NL-MPC but with slightly larger fluctuations in farm-level power output, highlighting the trade-off between performance and computational tractability. The control framework is among the first to address multi-region wind turbine dynamics together with market driven power tracking objectives for floating offshore wind farms. Such multi-region control becomes increasingly necessary in the floating turbine setting where large (region spanning) wind speed variations are common due to wave induced platform pitching.


\end{abstract}

\section{Introduction}

Wind is a highly variable resource that exhibits significant differences in speed and direction over temporal and spatial scales relevant to wind farm operation~\cite{2015_DraxlCliftonMcCaa_OverviewAndMeteorologicalValidationOfTheWindIntegrationNationalDataset}. Spatial variations can result in vast differences in the operational speeds for individual turbines within a single farm, whereas temporal shifts lead to fluctuations in each turbine's power output. Wind speeds are used to define three turbine operating regimes. In Region 1, wind speeds are lower than the cut-in speed. Region 2 comprises the cut-in speed to rated wind speed, while Region 3 spans the rated wind speed to the cut-out speed~\cite{2011PaoJohnson_ControlOfWindTurbines}. Turbine control approaches differ across  regions. In Region 1 the wind speed is insufficient to generate power, so the controller maintains a fixed rotor. Region 2  power maximization controllers typically adjust generator torque to track the maximum power coefficient \cite{2024_StockhouseZalkindRossPao_AnalysisOfPowerMaximizingRegion2Controllers}, whereas Region 3 control adjusts turbine blade pitch to regulate rotor speed and maintain rated (maximum) power \cite{2006_JohnsonPaoBalasFingersh_ControlOfVariableSpeedWindTurbines, 2022_Abbas_ROSCO}.

Wind farm control requires collective approaches for an array of turbines, whose individual control actions affect the wind speed experienced by downstream turbines in their wakes. Wake interactions render turbine centric control, such as turbine-level maximum power tracking (known as the greedy approach) suboptimal, i.e. greedy control produces lower overall farm power versus collective approaches \cite{2024_FrederikEtAl_ComparisonOfWindFarmControlStrategies_Turbine}. Accounting for wake interactions is of particular importance in power tracking applications (e.g. for wind farms providing secondary frequency regulation), where turbine set-points fall below the maximum power point and the desired trajectory changes at timescales similar to the inter-turbine travel time (propagation of the wake from one turbine to the next) ~\cite{2017Shapiro_ModelBasedRecedingHorizonControlOfWindFarmsForSecondaryFrequencyRegulation,2022Shapiro_TurbulenceAndControlOfWindFarms, 2016_FlemingAhoGebraadPaoZhang_CFDSimulationStudyOfActivePowerControl}. Farm-level power tracking therefore requires turbines to account for wake interactions, as well as possibly operate below their maximum power point (de-rate) to maintain the control authority required to increase the total farm power level in response to the desired trajectory. 

Some early farm-level power tracking research employed dynamic wake models to lower the de-rates necessary to achieve tracking performance goals~\cite{2017Shapiro_ModelBasedRecedingHorizonControlOfWindFarmsForSecondaryFrequencyRegulation}. This work and others using high-fidelity simulation models~\cite{2016GoitMeyers_OptimalCoordinatedControlOfPowerExtrationInLES} provided key insights regarding how real-time control actions affect wakes as they propagate through the farm. However, they used simplified representations of the turbines and turbine control and did not consider the different operational regions. Another line of research instead focused on turbine-level control design using simplified wake models ~\cite{2017_VanWingerdenPaoAhoFleming_ActivePowerControlOfWakedWindFarms}. 

Turbine control for farm-level power tracking has traditionally relied on reducing the rotor speed set-point in existing maximum power point tracking algorithms. In practice, this is often implemented using a lookup table, which is incorporated as a perturbation to the rated rotor speed reference \cite{2014_JeongJohnsonFleming_ComparisonTestingPowerReserveControlStrategies, 2014_ElaGevorgianFlemingEtAl_ActivePowerControlsFromWindPower}. Because Regions 2 and 3 use different control variables, a transition region (Region 2.5) blends the Region 2 and Region 3 controllers to promote continuity in turbine control set-points \cite{2022_Abbas_ROSCO} \cite{2019_ZalkindPao_ConstrainedWindTurbineControl}. This blending has typically employed set-point smoothing, where an offset from the reference is applied so the controller reaches the appropriate switching value more quickly. 

In this work, we emphasize the co-design of turbine and farm-level control for use in a floating offshore wind farm setting. We build upon the control paradigm introduced in\cite{2025_AmplemanAyalaGayme_TowardsCollectiveControlOfFloatingOffshoreWindFarms} for turbines operating in Region 3 in two ways. We first introduce a novel Region 2 controller specifically designed for the power tracking problem. Unlike conventional Region 2 control approaches that modify maximum power point tracking through lookup-table perturbations or gain scheduling, the proposed controller directly regulates rotor speed, allowing turbine operation to be coordinated with farm-level power tracking objectives. Through the regulation of rotor speed, the turbine is implicitly de-rated by the farm level controller to achieve the desired power setpoint. This contrasts with other approaches that explicitly de-rate turbines in order to allow sufficient actuation authority to achieve farm-level tracking objectives. We then propose a Region 2.5 control that smoothly transitions between the regions using a sigmoid function based weighting on the generator torque and blade pitch control commands, respectively used in regions 2 and 3. We implement the proposed turbine-level control as an inner-loop within two different outer-loop   wind farm controllers, (1) a Nonlinear MPC (NL-MPC) that includes a nonlinear model of the floating wind turbines and dynamic model of the wake evolution, and (2) a Linear Parameter Varying Time Delay MPC (LPVTD-MPC) based on an LPV model of the wind turbine and a delay model for the wake propagation time. These two approaches represent different levels of model fidelity and corresponding computational costs. We evaluate the wind farm power tracking performance for the two approaches based on signals from the PJM ancillary services regulation market using the certification criteria for market participation~\cite{2020_PJM_Manual}. 



The remainder of this paper is organized as follows. In Section \ref{sec:multi_region_turbine_control}, the multi-region turbine-level controller is introduced. Section \ref{sec:windfarm_MandC} includes a description of the (NL-MPC) and (LPVTD-MPC) wind turbine and farm-level models alongside their respective control approaches. The PJM RegD certification criteria and scoring algorithms are detailed in Section \ref{sec:PJM_Cert_Criteria}. In Section \ref{sec:Sim_and_Results} we present simulation results of the two control approaches applied to three different PJM regulation signals for an eight-turbine farm with different inlet velocity profiles obtained from large eddy simulation data~\cite{LESGO}. The two farm-level control algorithms are scored against the PJM certification criteria. The paper concludes and provides a brief discussion of directions for future work in Section \ref{sec:conclusion}.

\section{Turbine Control} \label{sec:multi_region_turbine_control}

This section introduces the proposed Region 2 and 3 control approaches together with the blending function for Region 2.5. For turbine-level control, we treat the rotor speed command $\omega_{c}$ as an input which may come from wind speed estimators, lookup tables, farm-level controllers, etc.

In Region 2, the collective blade pitch, $\beta_{c}$, is set to zero while the generator torque is controlled. The control design is based on the following first order model of rotor speed dynamics with a stiff drivetrain. 
\begin{equation}
    \dot{\omega}_{r} = \frac{1}{J_{eq}}\left(T_{a}(\omega_{r}) - N_{G}T_{g}\right),
    \label{eq:drive_dynamics}
\end{equation}
where $J_{eq}$ is the entire drivetrain moment of inertia and $N_{G}$ is the gear ratio between the high and low speed shafts, i.e., $\omega_{g} = N_{G}\omega_{r}$. $T_{a}(\omega_{r}) = P(\omega_{r})/\omega_{r}$ is the aerodynamic torque and $T_{g}$ is the generator torque which we wish to design. The aerodynamic power, $P(\omega_{r})$ is computed as
\begin{equation}
    P(\omega_{r}) = 0.5\rho_{a}A_{r}C_{p}(\lambda(\omega_{r}),\beta_{c} = 0)||u_{\infty}||^{3}
\end{equation}
where $\rho_{a}$ is the density of air, $A_{r}$ is the rotor swept area, $C_{p}$ is the power coefficient, and $\lambda = R_{r}\omega_{r}/||u_{\infty}||$ is the tip speed ratio with turbine blade radius $R_{r}$ and the inlet wind speed for the turbine, $u_{\infty}$. 

Assuming piecewise constant rotor speed commands, the error dynamics can be written in terms of the rotor speed dynamics (i.e. $\dot{e} = \dot{\omega}_{r}$) as
\begin{equation} \label{eq:error_dynamics}
    \begin{bmatrix} e \\ \dot{e} \end{bmatrix} = \begin{bmatrix} 0 & 1 \\ -K_{I_{T}} & -K_{p} \end{bmatrix}\begin{bmatrix}
        \int e \\ e
    \end{bmatrix},
\end{equation}
where $e=\omega_{r} - \omega_{c}$ is the error between a rotor speed command, $\omega_{c}$, and the rotor speed, $\omega_{r}$. The error dynamics are then
\begin{equation}
      \dot{\omega}_{r} = -K_{I_{T}}\int e -K_{p}e.
\end{equation}
Filling in \eqref{eq:drive_dynamics} and solving for the generator torque yields
\begin{equation} \label{eq:region2_controller}
 T_{g} = \frac{J_{eq}}{N_{G}}\left(K_{I_{T}}\int e +  K_{p}e\right) + \frac{P}{N_{G}\omega_{r}}.
\end{equation}
The expression for $T_{g}$ in \eqref{eq:region2_controller} can be interpreted as a feedback-linearization controller where the gains $K_{I_{T}}$ and $K_{p}$ can be specified to ensure the autonomous system in \eqref{eq:error_dynamics} is Hurwitz. Implementing \eqref{eq:region2_controller} requires power and rotor speed measurements. As in standard turbine control applications, we assume rotor speed is available. The power measurement is taken on the generator side as
\begin{equation} \label{eq:gen_fdbk}
    P_{gen} = \nu_{G} N_{G} \omega_{r} T_{g},
\end{equation}
where $\nu_{G}$ is the generator efficiency. Then, plugging in \eqref{eq:gen_fdbk} into \eqref{eq:region2_controller} and simplifying yields
\begin{equation} \label{eq:r2_fdbk}
    T_{g} = \frac{J_{eq}}{\left(1 - \nu_{G}\right)N_{G}}\left(K_{I_{T}}\int e +  K_{p}e\right).
\end{equation}
We can then form the error dynamics with respect to \eqref{eq:r2_fdbk} as
\begin{equation} \label{eq:closed_loop_error}
    \dot{e} = \frac{1}{J_{eq}}T_{a}(\omega_{r}) - \frac{1}{1 - \nu_{G}}\left(K_{I_{T}}\int e + K_{p}e\right).
\end{equation}
The error dynamics in \eqref{eq:closed_loop_error} have equilibrium points $(\int e,e) = (\int e^*,0)$ with $\int e^*$ defined as
\[\int e^* = (1 - \nu_{G})/(J_{eq}K_{I_{T}})T_{a}(\omega_{c})\]
Then the dynamics in \eqref{eq:closed_loop_error} can be shifted to the origin by defining $\int\tilde{e} = \int e - \int e^*$ and manipulated to obtain 
\begin{equation}
    \dot{e} = \frac{1}{J_{eq}}\left(T_{a}(\omega_{r}) - T_{a}(\omega_{c}\right) - \frac{1}{1 - \nu_{G}}\left(K_{I_{T}}\int \tilde{e} + K_{p}e\right)
\end{equation}
\begin{proposition}
     Given a bounded rotor speed  $\omega_{r}\in\Omega$, where $\Omega$ is a compact set that defines the rotor speed limits; if $K_{I_{T}} > 0$ and $K_{p} > \left(1 - \nu_{G}\right)L/J_{eq}$ where $L$ satisfies 
    \begin{equation} \label{eq:lipshitz_bound}
        ||T_{a}(\omega_{c} + e) - T_{a}(\omega_{c})||\le L||e||,
    \end{equation}
    then the closed loop Region 2 error dynamics \eqref{eq:closed_loop_error} are asymptotically stable for all $\omega_{r}\in\Omega$.
\end{proposition}
\begin{proof}
    Given the function
    \[V = \frac{1}{2}\left(e^2 + \frac{K_{I_{T}}}{1 - \nu_{G}}\int\tilde{e}^{2}\right),\] where $V>0$ for $(e,\tilde{e})\neq (0,0)$ for the given conditions on $K_{I_{T}}$ and $K_p$ and $\nu_{G}<1$. 
    The derivative can be computed as
    \[\dot{V} = \frac{e}{J_{eq}}\left(T_{a}(\omega_{c} + e) - T_{a}(\omega_{c})\right) - \frac{K_{p}}{1 - \nu_{G}}e^{2}.\]
    We can upper bound $\dot{V}$ by substituting $L$ from \eqref{eq:lipshitz_bound} to arrive at the condition that $K_{p} > \left(1 - \nu_{G}\right)L/J_{eq}$, which enforces $\dot{V} \le 0$. The set $\left\{\left(e,\int \tilde{e}\right):\dot{V}(e,\int \tilde{e}) = 0\right\}$ contains the point $e=0$ and the trivial trajectory $\int \tilde{e}=0$ since, by inspection of the dynamics in \eqref{eq:closed_loop_error}, $\dot{e} = 0$ only at the origin $\int \tilde{e}=0$. Thus, by LaSalle's invariance theorem, the origin is asymptotically stable. 
\end{proof}

Region 3 employs the collective pitch control approach proposed in\cite{2025_AmplemanAyalaGayme_TowardsCollectiveControlOfFloatingOffshoreWindFarms} to track a rotor speed command while maintaining stability of the floating turbine and its control loops. This is implemented as a PI-LQR controller because the turbine dynamics near rated operation are well approximated by a linear model, allowing the stabilization and disturbance-rejection properties of LQR to be used.  We select a linear model at one operating point to design penalty matrices $Q$ and $R$, which are then used to design integrator and state feedback gains, respectively $K_{I_{\beta}}$ and $K_{x}$, at evenly spaced wind speed values throughout Region 3: $\left\{12, 14, ..., 24\right\} \ [m/s]$. The matrices $Q$ and $R$ are designed to limit bandwidth below the frequency of the negative damping loop associated with the turbine thrust and platform pitch coupling (see e.g.,\cite{ 2007Larsen_AMethodToAvoidNegativeDampedLowFrequentToiwerVibrations, 2008Jonkman_InfluenceOfControlOnThePitchDampingOfAFloatingWindTurbine}). The resulting PI controller for collective blade pitch, $\beta_{c}$, is given by
\begin{equation} \label{eq:region3_controller}
    \beta_{c} = -K_{I_{\beta}}\int e  -K_{x}\chi_{fdbk}.
\end{equation}
Here $\chi_{fdbk}$ is a subset of states from the turbine model selected for feedback~\cite{2025_AmplemanAyalaGayme_TowardsCollectiveControlOfFloatingOffshoreWindFarms}, $K_{I_{\beta}}$ and $K_{\chi}$ are the Region 3 controller integral gain, and state feedback gains, respectively. The states in $\chi_{fdbk}$ typically include platform surge and pitch, their rates, as well as the turbine rotor and generator speeds.

The Region 2.5 control comprises a blending function that combines both collective blade pitch and torque control such that the turbine inputs are $s(u)\beta_{c}$ and $(1-s(u))T_{g}$. The relative weights are specified via a sigmoid function
\begin{equation} \label{eq:switching_function}
    s(u) = \frac{1}{1 + e^{-k_{s}\left(u - u_{0}\right)}}. 
\end{equation}

OpenFAST simulations were used to validate the approach and tune the sigmoid parameters $k_s$ and $u_0$ so that the collective blade-pitch and generator-torque controllers were activated at the appropriate transition speed. Region 2 gains in \eqref{eq:r2_fdbk} were also tuned to reduce controller-induced oscillations from additional plant dynamics (i.e. floating platform dynamics) that were not considered in the aerodynamic model used for controller design. Formally characterizing the performance and stability of this controller over Region 2.5 is an area of ongoing work.

\section{Wind farm controller designs} \label{sec:windfarm_MandC}

We now introduce the farm-level control approaches that will be used to generate the turbine rotor speed commands described in Section \ref{sec:multi_region_turbine_control}. We consider a floating wind farm arranged on a rectangular lattice of rows and columns, and assume the columns are spaced far enough apart that spanwise interactions are negligible (i.e., only streamwise interactions between turbines are considered). Model-based MPC is employed as it enables definition of a control horizon such that the wake interaction models can generate predictions over a horizon long enough to capture wake propagation across the farm. The MPC formulation also enables coordinated multivariable control of turbine set-points, while explicitly enforcing turbine operational constraints. Two  formulations are considered: a nonlinear MPC (NL-MPC) that uses a dynamic wake model to capture turbine interactions and a nonlinear model of the  turbine dynamics from \cite{2018HomerNagamune_PhysicsBased3DControlOrientedModelingofFloatingWindTurbines},  and a reduced-order LPVTD-MPC formulation that represents wake propagation via time delays and approximates the nonlinear turbine dynamics using linear parameter-varying models.  Comparing these two approaches highlights the trade-off between model fidelity and computational tractability for wind farm power tracking.

\subsection{Wind Farm Models} \label{sec:wind_farm_models}

\subsubsection{Turbine Model (NL-MPC)}
As in our prior work~\cite{2025_AmplemanAyalaGayme_TowardsCollectiveControlOfFloatingOffshoreWindFarms}, we employ the floating offshore turbine model from \cite{2018HomerNagamune_PhysicsBased3DControlOrientedModelingofFloatingWindTurbines}, which models the NREL semi-submersible Phase II OC4 5MW reference turbine  \cite{2007Jonkman_DynamicsModelingAndLoadsAnalysisOfAnOffshoreFloatingWindTurbine,2014Robertson_DefinitionOfTheSemiSubmersibleFloatingSystemForPhaseIIofOC4}. The equations of motion are given by
\begin{equation} \label{eq:turbine_dynamics}
    \dot{\mathbf{\chi}} = f(\mathbf{\chi},\mathbf{\eta},\mathbf{u}),
\end{equation}
where $f$ is a nonlinear function of the control input $\mathbf{\eta}$, the incoming wind $\mathbf{u}$, and the states 
\begin{equation*}
\mathbf{\chi} = 
        \begin{bmatrix} 
        \mathbf{r} & \Theta & \Delta\theta_{r} & \dot{\mathbf{r}} & \dot{\Theta} & \omega_{r} & \omega_{g} 
        \end{bmatrix}^{T}.
\end{equation*}
Here the vector $\mathbf{r}$ includes the translational degrees of freedom (DOFs), namely surge, sway and heave. Similarly ${\Theta}$ is the vector containing the three-dimensional rotational DOFs (platform roll, pitch and yaw). The quantities $\dot{\mathbf{r}}$ and $\dot{{\Theta}}$ denote their respective time derivatives, and $\Delta\theta_{r}=\theta_r-\frac{1}{N_{G}}\theta_g$ describes the rotor azimuth state.  As previously noted in Section \ref{sec:multi_region_turbine_control}, $\omega_{r}$ and $\omega_{g}$ are the respective rotational velocities of the rotor and generator. The three-dimensional wind velocity vector $\mathbf{u}$ is in the turbine reference frame. A complete description of the equations of motion is given in \cite{2018HomerNagamune_PhysicsBased3DControlOrientedModelingofFloatingWindTurbines}.

The control vector 
\begin{equation} \nonumber
    \mathbf{\eta} = \left[ s(u)\beta_{c} \quad (1-s(u))T_{g} \quad \gamma \right], 
\end{equation}
includes the collective pitch of the turbine blades, $\beta_{c}$, the generator torque, $T_{g}$, and the yaw angle of the turbine nacelle, $\gamma$. In closed loop, the turbine model introduced in Section \ref{sec:multi_region_turbine_control} uses \eqref{eq:region3_controller} to control $\beta_{c}$ and \eqref{eq:r2_fdbk} to control $T_{g}$ with $\gamma$ fixed to be aligned with the streamwise velocity axis. The sigmoid switching function \eqref{eq:switching_function} is used to distribute the relative influence of the  $\beta_{c}$ and $T_{g}$ control signals as the wind speed transitions between regions 2 and 3.

\subsubsection{Wake Advection (NL-MPC)}
The wake dynamics are modeled using the advection diffusion model introduced in \cite{2017Shapiro_ModelBasedRecedingHorizonControlOfWindFarmsForSecondaryFrequencyRegulation}, which has been utilized in applications ranging from induction factor based control \cite{2025_AmplemanAyalaGayme_TowardsCollectiveControlOfFloatingOffshoreWindFarms,2019Shapiro_AWakeModelingParadigmForWindFarmDesignAndControl} to thrust vectoring for turbine relocation via yaw\cite{2021_KheirabadiNagamune_FOFSymDyn}. The model takes the form
\begin{align} \label{eq:wake_dynamics}
    \frac{\partial{\delta u_{n}}}{\partial{t}} + &u_{\infty,n}(x,t)\frac{\partial{\delta u_{n}}}{\partial{x}} = \nonumber \\
                                                    &-w_{n}(x,t)\delta u_{n}(x,t) + S_{n}(t)G(x - s_{x,n})
\end{align}
where $t$ is time, $x$ is the streamwise coordinate and $\delta u$ is the wake deficit caused by turbine $n$ which is then propagated downstream. Here we employ the hub height velocity at the inlet to turbine $n$, denoted $u_{\infty,n}$, as the advection velocity. The forcing term $S_{n}(t)$ is a function of the turbine thrust coefficient, $C_{T,n}$, and the inlet velocity, and $G(x - s_{x,n})$ is a Gaussian shaping function whose argument depends on the streamwise location of turbine $n$, denoted $s_{x,n}$. The boundary conditions include the initial velocity deficit at each turbine and the corresponding time derivative of the velocity deficit at the turbine location $x=s_{x,n}$, which are given by
\begin{equation}
    \delta u_{n}\rvert_{x = s_{x,n}} = u_{\infty,n}(1 - \sqrt{(1 - C_{T,n})}),
\end{equation}
\begin{equation}
    \frac{\partial{\delta u_{n}}}{\partial{t}} \biggr\rvert_{x = s_{x,n}} = 0.
\end{equation}
The streamwise velocity field is then calculated via
\begin{equation} \label{eq:velocity_field}
    u(\mathbf{x},t) = U_{\infty,p}(t) - \sum_{n} \hat{n}(t)\delta u_{n}(x,t)W_{n}(x),
\end{equation}
where $W_n(x)$ is a super-Gaussian wake shape function detailed in \cite{2019Shapiro_AWakeModelingParadigmForWindFarmDesignAndControl}, $U_{\infty,p}(t)$ is the incoming velocity to the turbine furthest upstream for the given column $p$ (i.e., at the front of the wind farm with respect to the turbine of interest). The summation represents superposition of the upstream wakes, and $\hat{n}(t)=\mathbf{v}_{n}/||\mathbf{v}_{n}||$ is a unit vector that describes the direction of the incoming wind with respect to a vector normal to the hub of the turbine\cite{2018HomerNagamune_PhysicsBased3DControlOrientedModelingofFloatingWindTurbines}. 
\subsubsection{Turbine Model (LPVTD-MPC)} We chose to evaluate an
LPV model of the turbine as a low order approximation due to its ability to closely approximate the nonlinear turbine dynamics by linearizing the model at a variety of operating points. Similar to the approach in \cite{2025_AmplemanAyalaGayme_TowardsCollectiveControlOfFloatingOffshoreWindFarms}, the nonlinear model in \cite{2018HomerNagamune_PhysicsBased3DControlOrientedModelingofFloatingWindTurbines} was linearized at a range of operating points to capture the turbine dynamics in closed loop with \eqref{eq:r2_fdbk} and \eqref{eq:region3_controller}  in the regions of interest. Specifically, the array of linear models was generated by capturing rotor speed and wind speed magnitude at even increments throughout the operational range: $\left\{8, 8.5, ..., 12\right\} \ [rpm]$ and $\left\{8, 10, ..., 24\right\} \ [m/s]$. Each turbine $i$ was then represented by an LPV model of the following continuous time form:
\begin{equation} \label{eq:LPV_turbine_dynamics}
    \dot{\chi}_{cl,i} = A(\omega_{r,i},||u_{\infty,i}||)\chi_{cl,i} + B(\omega_{r,i},||u_{\infty,i}||)\psi_{cl,i}
\end{equation}
\begin{equation}
    \zeta_{cl,i} = C(\omega_{r,i},||u_{\infty,i}||)\chi_{cl,i} + D(\omega_{r,i},||u_{\infty,i}||)\psi_{cl,i},
\end{equation}
where $\psi_{cl,i}^{k} = \begin{bmatrix} \omega_{c,i} &  u_{\infty,x_{i}} & u_{\infty,y_{i}} & u_{\infty,z_{i}} \end{bmatrix}^{T}$ contains the rotor-speed command from the farm-level controller together with the wind inlet velocity for turbine $i$. The subscript ``$cl$'' on $\chi, \psi$, and $\zeta$ indicate that they are the states, inputs, and outputs for the closed loop LPV turbine-level model.

\subsubsection{Wake Advection (LPVTD-MPC)}

A model similar to those in \cite{2025_AmplemanAyalaGayme_TowardsCollectiveControlOfFloatingOffshoreWindFarms} \cite{2021Starke_NetworkBasedEstimationofWindFarmPowerandVelocityDataUnderChangingWindDirections} is used to propagate velocity deficits throughout the farm. The deficit from turbine $i$ to a downstream turbine $j$ separated by a streamwise distance $\varkappa {(i,j)}$ is given by 
\begin{equation} \label{eq:initial_deficit}
    \delta u_{0,i} = u_{\infty,i}\left(1 - \sqrt{1 - C_{T,i}}\right)
\end{equation}
\begin{equation}
    d_{w,(i,j)} = 1 + k_{w,i} \ln\left( 1 + \exp(\varkappa{(i,j)}/(R\sqrt{2}))\right)
\end{equation}
\begin{equation}
    \eta_{\alpha,j}^{k} = \sum_{i\in\mathcal{G}_{j}}\frac{\delta u_{0,i}}{d_{w,(i,j)}^{2}}\left(1 + erf(\varkappa{(i,j)}/(R\sqrt{2}))\right),
\end{equation}
where $d_{w,(i,j)}$ is the wake diameter and $k_{w}$ is the wake growth coefficient. The summation over $\mathcal{G}$ includes all turbines which impart a velocity deficit affecting turbine $j$. This information is then propagated using a first order time delay model 
\begin{equation} \label{eq:time_delay_advection}
    \chi_{\alpha}^{k+1} = \mathcal{A}^{k}\chi_{\alpha}^{k} + \left(\mathcal{A}^{k} - 1\right)\frac{\tau^{k}}{2}\nu_{\alpha}^{k}
\end{equation}
\begin{equation}
    \zeta_{\alpha}^{k} = \frac{4}{\tau^{k}} \chi_{\alpha}^{k} - \nu_{\alpha}^{k}
\end{equation}
\begin{equation}
    u_{\infty,j}^{k} = U_{\infty}^{k} - \zeta_{\alpha}^{k}.
\end{equation}
\begin{equation} \label{eq:advection_update}
    U_{a,j}^{k+1} = u_{\infty,j}^{k}\sqrt{1 - C_{T,j}^{k}}.
\end{equation}
where $\mathcal{A}^{k} = \exp\left(-\Delta t/2\tau^{k}\right)$ and $\tau^{k} = \varkappa{(i,j)}/||U_{a,i}^{k}||$. Here, the subscript ``$\alpha$'' on $\chi, \nu$, and $\zeta$ indicate that they are associated with the velocity deficit time delay system dynamics. 

\subsection{Farm-Level Controller} \label{sec:farm_level_controller}

The farm-level controller computes optimal rotor speed commands for each turbine for each farm-level control sample period (i.e., at the specified update rate). Commands are optimized over $N_{C}$ control horizons, each of duration $T_{C}$. These parameters are designed so that the prediction horizon $(N_{C}T_{C})$ captures turbine settling time and the duration necessary for wakes to propagate from the front to the back of the farm. The cost function for farm-level power tracking contains  a quadratic power tracking error term and a quadratic rotor speed regularization term, and is defined as
\begin{align} \label{eq:horizon_cost}
J_{q} = &Q_{\omega}\sum_{i=1}^{N_{T}}\left(b_{\omega}^{T}\chi_{i}(t_{q}) - \omega_{c,i}^{q}\right)^{2} + \nonumber \\ & \int_{t_{q}}^{t_{q+1}}\left\{Q_{e}\left(P_{sp}(\tau) - \sum_{i=1}^{N_{T}} b_{P}^{T}\zeta_{i}(\tau)\right)^{2}\right\}d\tau,
\end{align}
where $N_{T}$ is the number of turbines in the wind farm and each $q\in \left\{1,..,N_{C}-1\right\}$ represents one control horizon and $P_{sp}$ is the farm level power setpoints from the grid service market operator. The regularization term seeks to minimize the difference between the command and the rotor speed at the beginning of each horizon. This is done in an effort to maintain continuity in commands as well as to improve convexity of the cost function. The variables $b_{P}$ and $b_{\omega}$ contain all zeros except for the index corresponding to the variable of interest (i.e. power or rotor speed). $Q_{e}$ and $Q_{\omega}$ are weights for the respective tracking error and regularization term. The optimization problems for the two farm-level control schemes are given by
\begin{align}
\tag{NL-MPC}
\bm{\omega}_{c} = \arg\min_{\bm{\omega}_{c}^{q}\in\Omega} \quad &    \sum_{q = 1}^{N_{C}-1}J_{q}(t,\omega_{c}^{q}) \nonumber\\ 
\textrm{s.t.} \quad & \eqref{eq:turbine_dynamics}, \ \eqref{eq:wake_dynamics}   \nonumber \\
\nonumber \\
\tag{LPVTD-MPC}
\bm{\omega}_{c} = \arg\min_{\bm{\omega}_{c}^{q}\in\Omega} \quad &     \sum_{q = 1}^{N_{C}-1}J_{q}(t,\omega_{c}^{q}) \nonumber\\
\textrm{s.t.} \quad & \eqref{eq:LPV_turbine_dynamics}, \ \eqref{eq:time_delay_advection}, \ \eqref{eq:advection_update}   \nonumber
\end{align}
where $q = \left\{1,...,N_{C} - 1\right\}$ and $\Omega:=\left\{\omega_{c}:\omega_{c}\in\left[8 \ 12\right] \text{rpm}\right\}$. After each optimization problem is solved, only the first horizon's command is used (i.e., $\bm{\omega}_{c}^{q=1}$).

\section{Frequency Regulation in the PJM Market}

The cascaded wind-farm control strategies described in Sections \ref{sec:multi_region_turbine_control} and \ref{sec:windfarm_MandC} are evaluated against certification scoring criteria used by regional transmission operator (RTO) PJM for their RegD frequency regulation market (see Section 4 of \cite{2020_PJM_Manual})\footnote{During the review process of this paper, the criteria for participation in PJM regulation markets were in flux. Here, we use the recently (10/2025) archived criteria and signal definitions.}. These criteria are described in Section \ref{sec:PJM_Cert_Criteria} and then used to evaluate control performance for three 40 minute test signals in Section \ref{sec:Sim_and_Results}.

\subsection{PJM RegD Certification Criteria} \label{sec:PJM_Cert_Criteria}
PJM market qualification criteria uses an aggregate measure of precision, delay, and correlation scores. The precision score is computed over the entire test duration as
\begin{equation}
    S_{P} = 1 - \frac{1}{N}\sum_{i = 1}^{N}\frac{|P^{i}_{gen} - P^{i}_{sp}|}{\overline{P_{sp}}},
\end{equation}
where $P_{sp}$  denotes the power setpoint,  $P_{gen}$ is the raw output power and $\overline{P_{sp}}$ is the setpoint averaged over the test period. 
The delay and correlation scores are respectively computed together as
\begin{equation}
    S_{D}(\delta) = \left|\frac{\delta - T_{W}}{T_{W}}\right|,
\end{equation}
\begin{equation} \label{eq:corr_coeff}
    S_{C}(x,y) = \frac{\frac{1}{N}\sum_{i = 1}^{N}\left(x_{i} - \bar{x}\right)\left(y_{i} - \bar{y}\right)}{\sigma_{x}\sigma_{y}}. 
\end{equation}
These values are a function of the sub-interval $\delta$ within the five minute window $T_{W} = 300 \ [sec]$ that maximizes the sum of the two scores, i.e.,
\begin{equation} \label{eq:delta_star}
\begin{aligned}
\delta^{*} = \arg\max_{\delta} \quad &     S_{D}(\delta) + S_{C}(P_{gen},P_{sp})|_{t_{0} + \delta}^{t_{f}+\delta}. 
\end{aligned}
\end{equation}
Once $\delta^{*}$ is determined, the aggregate score for that five minute interval is computed as
\begin{equation}
    S = \frac{1}{3}\left(S_{D}(\delta^{*}) + S_{C}(P_{gen},P_{sp})|_{t_{0} + \delta^{*}}^{t_{f}+\delta^{*}} + S_{P}\right).  
    \label{eq:scoring}
\end{equation}
The $S_D$ and $S_C$ signals are sampled every 10 seconds and scoring is performed over five minute intervals. These interval scores are then averaged over the full regulation period to obtain a composite score for the test. Certification requires an aggregate score  of 0.75 or better for three consecutive tests of at least 40 minutes in duration.

\subsection{Simulation Results} \label{sec:Sim_and_Results}

 We perform simulations with three different regulation signals comprising the standard PJM RegD testing signal and two 40 minute signals from the RegD archive \cite{PJM_Ancillary_Services}. All signals from \cite{PJM_Ancillary_Services} are normalized between $[-1,1]$, here we scale the regulation  to have a range of  $\pm$3MW and impose a  mean of 30MW so that the regulation signal is 10\% of the mean. An inlet wind profile for each column is obtained from the precursor domain from a large eddy simulation (LES) of a wind farm using the LESGO code \cite{LESGO}. Each three dimensional velocity component ($u$, $v$, $w$) is calculated by performing disk averaging over the turbine rotor for the first turbine in the given turbine column. Three wind inlet signals are generated by scaling by the original $u$, $v$, $w$ components. 
 The  wind inlet signals and RegD tracking signals are shown in Figures \ref{fig:PJM_wind_speeds} and \ref{fig:regulation_signals}, respectively. Here the different test signals and the associated wind inlet are indicated by color. For Scenario A the wind scale factor is 1 for both columns.  Scenario B wind data has a scale factor of 1.1 for column 1 and 1.0 for column 2, while Scenario C wind data has a scale factor of 1.0 for column 1 and 1.2 for column 2.

 \begin{figure}[ht] 
    \centering
    \includegraphics[scale=0.45]{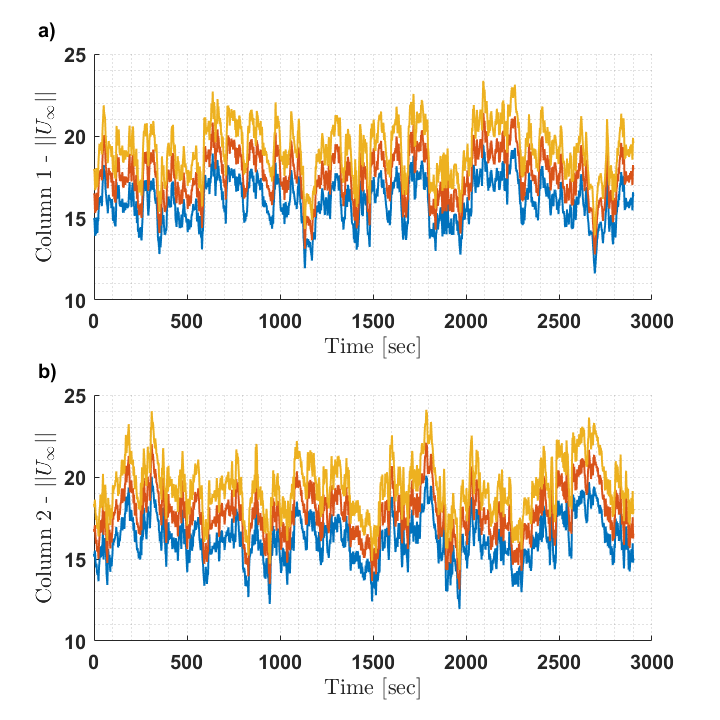}
        \vspace{-25pt}
    \caption{Three wind scenarios are generated by scaling two LES inlet wind profiles by factors $\left\{1.0, 1.1, 1.2\right\}$. Scenario A uses $\left\{1.0,1.0\right\}$ blue in a), b), Scenario B uses $\left\{1.1,1.0\right\}$ orange in a), blue in b), and Scenario C uses $\left\{1.0,1.2\right\}$ blue in a), yellow in b). This incoming wind is in Region 3, so the first turbine in each column sees Region 3 wind speeds but wake effects push the remaining turbines into regions 2 and 2.5.}
\label{fig:PJM_wind_speeds}
\end{figure}

\begin{figure}[ht] 
    \centering
    \includegraphics[scale=0.45]{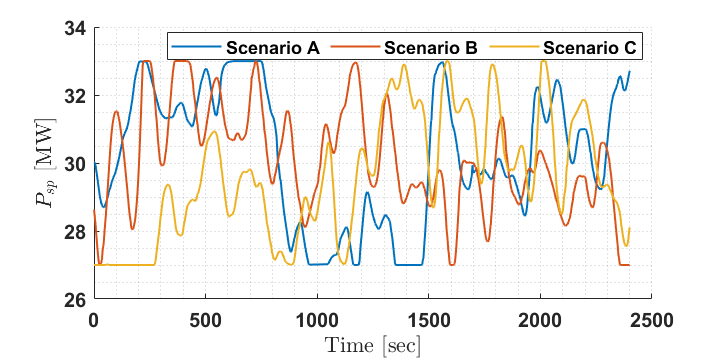}
            \vspace{-20pt}
    \caption{The three test signals from \cite{PJM_Ancillary_Services}, which have been scaled by factor of 3MW and then added to a mean demand of 30MW. }
\label{fig:regulation_signals}
\end{figure}

The (NL-MPC) and (LPVTD-MPC) farm-level controllers introduced in Section \ref{sec:windfarm_MandC} are applied to an eight-turbine wind farm arranged in four rows and two columns. The plant dynamics are described by \eqref{eq:turbine_dynamics} and \eqref{eq:wake_dynamics}. The turbine level controller gains were tuned in OpenFAST based on the NREL semi-submersible Phase II OC4 5MW reference turbine.
The turbine states and inflow wind are assumed available for feedback at both turbine and farm control levels. The MPC control horizon duration was set to $T_C = 50,[\text{sec}]$ with $N_C = 5$ intervals, with weights $Q_{e}=1$ and $Q_{\omega}=0.01$. The values for the Region 2.5 sigmoid function were set to $k_{s} = 5$ and $u_{0} = 11.5 \ [m/s]$. Each simulation runs for 2900 seconds with the first 500 seconds used as farm startup time. 
 
Matlab's \textit{fmincon} function with the interior-point algorithm enabled was used to solve the optimization problems (NL-MPC) and (LPVTD-MPC). Simulations were run on a standard Dell XPS 15 PC with parallel processing for finite difference gradient computation during optimization using 14 cores. The average solve time was 945 seconds for (NL-MPC) and 465 seconds for (LPVTD-MPC). Although (LPVTD-MPC) was able to reduce the amount of time to complete an MPC iteration by nearly a factor of 2 when compared to (NL-MPC), additional gains in processing speed are necessary for real-time control applications.

Figure \ref{fig:all_scenarios_tracking} shows the tracking performance for the two control approaches, with the aggregate PJM score indicated in the legend. Wind farm power tracks the reference signal well in all cases, with some fluctuations in power output likely due in some part to the turbulence in the wind inflow seen in Figure \ref{fig:PJM_wind_speeds}. Both (LPVTD-MPC) and (NL-MPC) achieve passing certification scores of at least 89.9\% for all scenarios, which is well above the 0.75 PJM market qualification requirement. In Scenarios A, B, and C, the precision scores for $\{\text{(NL-MPC) }, \text{(LPVTD-MPC)}\}$ were $\{0.986, 0.979\}$, $\{0.986, 0.981\}$, and $\{0.984, 0.979\}$, respectively. These precision scores of nearly 1 indicate that the tracking error was low with respect to the average value of the setpoint signal.

Figure \ref{fig:PJM_Scoring}a) shows the delay and correlation scores for each scenario. The delay scores were fairly high throughout all of the tests, with the highest combined correlation and delay scores associated with small $\delta^{*}$ \eqref{eq:delta_star}. The correlation scores were the lowest, with (LPVTD-MPC) consistently scoring lower than (NL-MPC), likely due to the larger power fluctuations visible in Figure \ref{fig:all_scenarios_tracking}.

As noted above, power-output fluctuations naturally arise due to the turbulent nature of the incoming wind. The same turbulent inflow, together with the control induced variations, may increase fatigue loading on the turbine blades and structures, affecting turbine performance and operational lifespan. Evaluating these effects and including fatigue mitigation as a control goal is an area of interest of future work but is outside the scope of this paper.

\begin{figure}[ht] 
    \centering
    \includegraphics[scale=0.45]{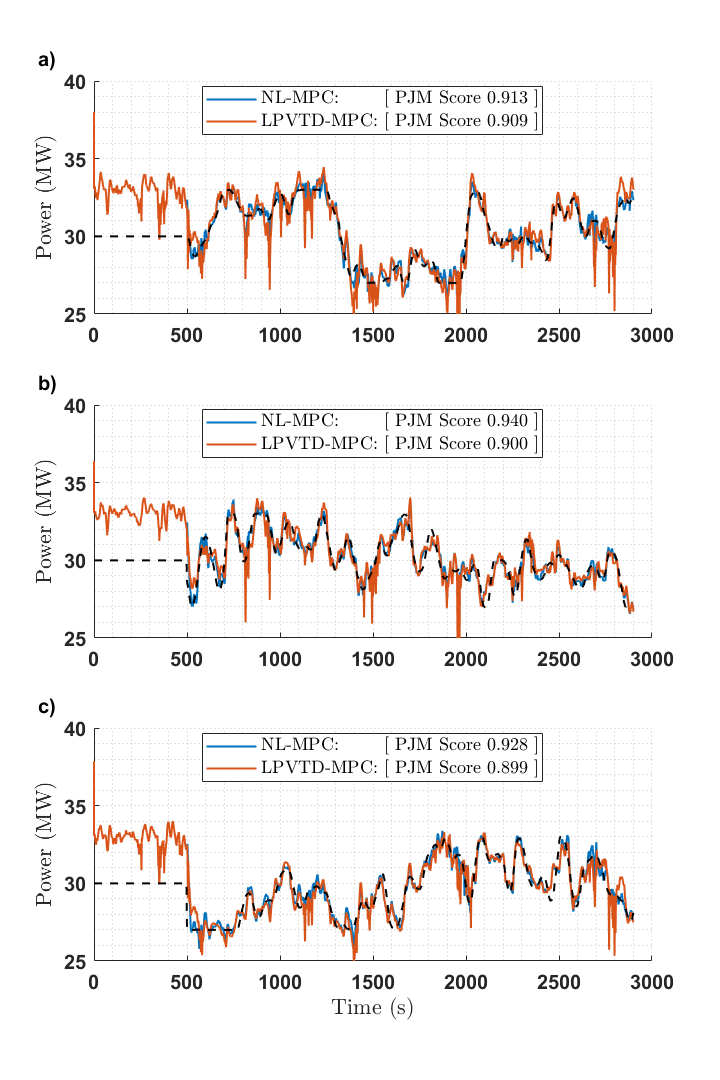}
        \vspace{-40pt}
    \caption{Controller performance of both NL-MPC and LPVTD-MPC farm controllers for a) Scenario A, b) Scenario B and c) Scenario C. For all cases farm-level control begins after 500 seconds. Both controllers achieve composite scores above the qualification threshold. As expected, LPVTD-MPC performs slightly worse than NL-MPC, and also appears to be more sensitive to disturbances.}
    \label{fig:all_scenarios_tracking}
\end{figure}
\begin{figure}[ht] 
    \centering
    \includegraphics[scale=0.45]{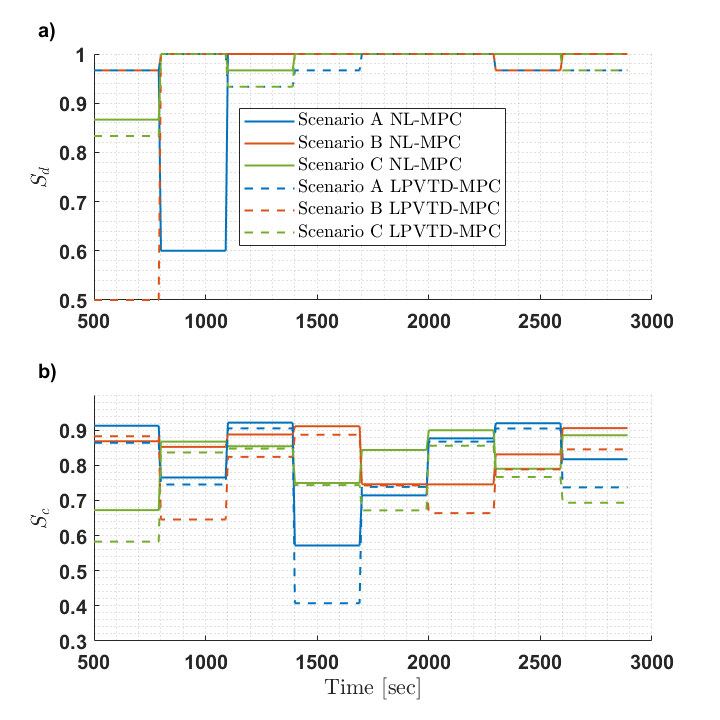}
            \vspace{-25pt}
    \caption{PJM Scoring for all three scenarios is shown. In a) the precision score is shown for each scenario to be near 1. In b) the delay score is lower initially in some cases but improves over time. In c) the correlation score is shown to have the most variability, with LPVTD-MPC consistently scoring lower than NL-MPC.}
\label{fig:PJM_Scoring}
\end{figure}

\section{Conclusion} \label{sec:conclusion}
This paper describes the integration of a multi-region turbine-level controller within farm-level power tracking algorithms. Region 2 control is accomplished through the use of a novel generator torque control law which smoothly progresses into a Region 3 controller \cite{2025_AmplemanAyalaGayme_TowardsCollectiveControlOfFloatingOffshoreWindFarms} based on a sigmoid transition function. This controller is integrated into two closed-loop floating offshore wind farm control algorithms for power tracking over multiple wind speed regions. The two approaches are shown to achieve power tracking that meets PJM RegD ~\cite{PJM_Ancillary_Services} certification criteria, earning scores of at least 89.9\% for all three scenarios, which greatly exceeds the minimum threshold of 75\%.

There are several avenues for future work. The approach presented here assumes perfect sensor feedback, which will never be achieved in practice. Future work will evaluate the impact of uncertainty, specifically with respect to the preview wind information \cite{2022_RobeyLundquist_BehaviorAndMechanismsOfDopplerWindLidarError}. The present work took an important step toward real-world implementation constraints by not assuming the entire PJM RegD signal is available at each MPC iteration, and instead using only the setpoint available at that time instant to inform the optimal rotor speed commands over the prediction horizon. However, this required more frequent MPC iterations to maintain sufficient tracking accuracy, leading to increased computational costs that would inhibit real-time implementation. Future work will investigate faster gradient computation methods for the farm-level MPC optimization problem, such as adjoint methods, that have been used in similar contexts~\cite{2018Meyers_DynamicStrategiesForYawAndInductionControlOfWindFarms,2022_VanDenBoekVanWingerden_AdjointOptimisationForWindFarmFlowControl}. At the turbine level, future work will address controller inefficiencies associated with rapid switching between operating regions as wind speed varies.

\bibliographystyle{ieeetr}
\bibliography{windfarmcontrol}

\end{document}